\documentclass[11pt]{article}

\usepackage{amsmath,amsthm,amssymb,amsfonts,verbatim}
\usepackage{latexsym}
\usepackage{fullpage}
\usepackage{graphicx}
\usepackage{float}

\usepackage[dvipsnames,usenames]{color}
\usepackage[pdftex,letterpaper=true,colorlinks=true,pdfpagemode=none,urlcolor=blue,linkcolor=blue,citecolor=BrickRed,pdfstartview=FitH]{hyperref}
\usepackage{hyperref}
\usepackage{times}

 \providecommand{\F}{\mathbb{F}}

\date{}

\newtheorem{lemma}{Lemma}
\newtheorem{theorem}{Theorem}

\newtheorem{prop}[lemma]{Proposition}

\newtheorem{defn}{Definition}

\newtheorem{rmk}{Remark}

\newtheorem{ex}{Example}

\def \mC {\mathcal{C}}
\def \mA {\mathcal{A}}

\def \mA {\mathcal{A}}
\def \mB {\mathcal{B}}
\def \mC {\mathcal{C}}

\def \mP {\mathcal{P}}

\def \mS {\mathcal{S}}

\def \Xi {{X^{[i]}}}

\newcommand{\Ga}{\alpha}

\newcommand{\Gg}{\gamma}     
\newcommand{\Gd}{\delta}

\newcommand{\poly}{\rm{poly}}

\def \bb {{\bf b}}

\def \bc {{\bf c}}

\def \bx {{\bf x}}

\def \by {{\bf y}}

\def \bv {{\bf v}}

\def \wt {{\rm wt}}
\def \mS {{\mathcal S}}
\begin{document}

\title{List Decodability of Symbol-Pair Codes}

\author{Shu Liu\thanks{ Shu Liu was with the National Key Laboratory of Science and Technology on Communications, University of Electronic Science and Technology of China, Chengdu 611731, China (email: shuliu@uestc.edu.cn)},  \; Chaoping Xing\thanks{ Chaoping Xing was with Division of Mathematical Sciences, School of Physical \&  Mathematical Sciences, Nanyang Technological University,  Singapore 637371 (email: xingcp@ntu.edu.sg)}\; and\; Chen Yuan\thanks{Chen Yuan was with Centrum Wiskunde \& Informatica, Amsterdam, Netherlands (email: Chen.Yuan@cwi.nl) Part of this work was done when the author was with the School of Physical and Mathematical Science, Nanyang Technological University, Singapore.}}

\maketitle
\begin{abstract}
We investigate the list decodability of symbol-pair codes in the present paper. Firstly, we show that list decodability of every symbol-pair code does not exceed the Gilbert-Varshamov bound. On the other hand, we are able to prove that with high probability, a random symbol-pair code can be list decoded up to the Gilbert-Varshamov bound. Our second result of this paper is to derive the Johnson-type bound, i.e., a lower bound on list decoding radius in terms of minimum distance. Finally,   we present a list decoding algorithm of Reed-Solomon codes beyond the Johnson-type bound.

\iffalse
In this paper, we show that the list decodability of the symbol-pair codes does not exceed the Gilbert-Varshamov bound regardless of alphabet size $q$. When $q$ is sufficiently large, the asymptotic Gilbert-Varshamov bound coincides with the Singleton bound for symbol-pair codes. Moreover, with high probability, a random symbol-pair codes can be list decoded up to the Gilbert-Varshamov bound with polynomial list size. When $q=\omega(1),$ list decoding radius can achieve the Singleton bound. Furthermore, by applying the existing decoding algorithms for Hamming metric codes, we present explicit constructions of symbol-pair codes that can be efficiently list decoded up to the Singleton bound.
\fi

\end{abstract}
\section{Introduction}

The high-density data storage technologies aim at designing the high-capacity storages at a relatively low cost. To achieve this goal, the theory of symbol-pair coding \cite{Y.C2011} was proposed to handle channels that output pairs of overlapping symbols, rather than one symbol at a time. Such channels, so called \textit{symbol-pair read channels}, introduce a new metric called \textit{pair distance}. It was showed that the pair error correcting capability of a code is larger than the error correcting capability of the same code in the Hamming metric. Cassuo and Litsyn~\cite{CL} gave an asymptotic lower bound on coding rates. This lower bound also indicates the existence of symbol-pair codes with higher rate than the codes in Hamming distance provided that both codes have the same relative distance.
 Chee et al.~\cite{YM2013} established a Singleton-type bound and showed the existence of symbol-pair codes meeting this bound.
 Following this direction, several works contributed to the constructions of symbol-pair codes meeting this bound~\cite{BGJ} and \cite{Kai}.

In this paper, we focus on the list decoding of symbol-pair codes.
This concept of list decoding was first introduced by Elias~\cite{P.E1957} and Wozencraft~\cite{J.M.W1958}.
Unlike the unique decoding algorithm, list decoding algorithm outputs a list of candidate codewords so as to tolerate
and correct more errors.
One of the key issues in coding theory is to explicitly construct codes with large list decoding radius.
Since there are too many works concerned with this topic, we refer the reader to \cite{thesis} for details.
Inspired by the list decoding in Hamming metric, we establish the lower bound and upper bound on the list decoding radius
of symbol-pair codes. We also reveal the differences between the codes in Hamming metric and symbol-pair metric by observing the different behaviours of the list decoding of Reed-Solomon codes in both metrics.

%We investigate the list decodability of the symbol-pair codes and then study the list decoding of random symbol-pair codes. Our results consist of three parts. Firstly, we show that the list decodability of symbol-pair codes does not exceed the Gilbet-Varshamov bound. On the other hand, when the field size $q$ is sufficiently large, the limit of list decodability of symbol-pair codes is the Singleton bound. Secondly, we show that, with high probability, a random ($\F_q$-linear) symbol-pair code can be list decoded up to the Singleton bound with $q=\omega(1).$ Thirdly, by applying the existing decoding algorithms for Hamming metric codes, one can efficiently list decode symbol-pair codes up to the Singleton bound if the alphabet size depends on $n.$

\subsection*{Previous results}
There are many works dedicated to unique decoding of symbol-pair codes.
Cassuto and Blaum \cite{Y.C2011} presented their decoding algorithm based on the error decoding algorithm in the Hamming metric. Yaakobi, Bruck and Siegel gave two constructions of effective decoding algorithms for linear cyclic codes~\cite{Yaa2012} and~\cite{Yaa2016}.
The decoding algorithm utilizing the syndrome of symbol-pair codes was proposed in~\cite{Masanori2014} by Hirotomo, Takita and Morii. They~\cite{Makoto2016} subsequently give an error-trapping decoding algorithm that is required to impose some restrictions on the pair error patterns.
There is a decoding algorithm based on linear programming designed for binary linear symbol-pair codes in~\cite{Shunsuke2016} by Horii, Matsushima and Hirasawa.
%To fully exploit the pair distance of symbol-pair code, we need to resort to the list decoding algorithm which can decode the codes beyond the unique pair distance.

\subsection*{Our results}
To the best of our knowledge, all known decoding algorithms are designed for the unique decoding of symbol-pair codes.
In this paper, we investigate the list decoding of symbol-pair codes. We first establish the Gilbert-Varshamov bound as an upper bound on the list decoding radius for all the symbol-pair codes. On the other hand, we also show that most random symbol-pair codes can be list decoded up to this bound. Then, we derive the Johnson-type bound in terms of minimum  distance which indicates that any symbol-pair codes can be list decoded beyond this bound. To show tightness of this bound, we further construct symbol-pair codes that can not be list decoded slightly beyond this bound, while it is an open problem whether there exists any Reed-Solomon code list decodable beyond the Johnson-type bound in Hamming metric. Finally, we  give an explicit  list decoding algorithm for a family of Reed-Solomon codes beyond this Johnson-type bound.

\subsection*{Organization}
This paper is organized as follows. In Section $2$, we introduce definitions of  symbol-pair codes, the Gilbert-Varshamov bound and some preliminaries on list decoding. In Section $3$, we establish an upper bound on the list decoding radius of symbol-pair codes, i.e., the Gilbert-Varshamov bound. In addition, in Section $3$ we also show that, with high probability, a random code can be list decoded up to the Gilbert-Varshamov bound. The Johnson-type bound is derived in Section $3$ as well.  In Section $4$, we present an list decoding algorithm of Reed-Solomon codes beyond the Johnson-type bound.

\section{Preliminaries}

Let $q$ be the finite field with $q$ elements, where $q$ is a power of a prime, and let $\F_q^n$ denote the set of all vectors of length $n$ over $\F_q.$ The Hamming weight of $\mathbf{x}$ is denoted by $\wt_{\sf H}(\mathbf{x}).$ A $q$-ary Hamming metric code $\bf C$ of length $n$ is a subset of $\F_q^n.$ The code $\bf C$ is called $(\tau n,L)_{\sf H}$-list decodable if for every word $\mathbf{y}\in\F_q^n,$ the intersection of $\bf C$ with the Hamming ball $\{\mathbf{x}\in\F_q^n: \wt_{\sf H}(\mathbf{x}-\mathbf{y})\leq \tau n\}$ has size at most $L,$ here the parameter $L$ is called the list size.

Then, we move to introduce the definitions of symbol-pair codes.
%%%%%%%%%%%%%%%%%%%%%%%%%%%%%%%%%%%%%%%%%%%%%%%%%%%%%%%
\begin{defn}(Symbol-pair Read Vector) Let $\mathbf{x}=[x_0, x_1, \cdots, x_{n-1}]$ be a vector in $\F_q^n.$ The symbol-pair read vector of $\mathbf{x}$ is defined as
\begin{eqnarray*}
\pi(\mathbf{x})=[(x_0,x_1), (x_1,x_2),\cdots,(x_{n-2},x_{n-1}),(x_{n-1},x_0)].
\end{eqnarray*}
%every vector $\mathbf{x}\in\F_q^n$ has a pair representation $\pi(\mathbf{x})\in\left(\F_q\times\F_q\right)^n.$
\end{defn}
%%%%%%%%%%%%%%%%%%%%%%%%%%%%%%%%%%%%%%%%%%%%%%%%%%%%%%%
The pair distance between two vectors in $\F_q^n$ is the Hamming distance between their corresponding pair vectors, where two pairs $(a,b)$ and $(c,d)$ are viewed as different if either $a\neq c$ or $b\neq d$.
\begin{defn}(Pair Distance) Let $\mathbf{x}=(x_0, x_1, \cdots, x_{n-1})$ and $\mathbf{y}=(y_0, y_1, \cdots, y_{n-1})$ be two vectors in $\F_q^n.$ The pair distance between $\mathbf{x}$ and $\mathbf{y}$ is defined as
\begin{eqnarray*}
d_{\sf P}(\mathbf{x},\mathbf{y})&=&d_{\sf H}(\pi(\mathbf{x}),\pi(\mathbf{y}))\\&=&|\{0\leq i\leq n-1: (x_i, x_{i+1})\neq (y_i, y_{i+1})\}|.
\end{eqnarray*}
\end{defn}
The pair weight of a vector $\bx\in \F_q^n$ is defined as $\wt_{\sf P}(\bx)=d_{\sf P}(\bx,\mathbf{0})$ where
$\mathbf{0}$ is the all-zero vector of $\F_q^n$.
%%%%%%%%%%%%%%%%%%%%%%%%%%%%%%%%%%
The minimum pair distance of a code $\mC\in\F_q^n$ is defined as
\begin{eqnarray*}
d_{\sf P}(\mC)=\mathop{\min}\limits_{\mathbf{x}, \mathbf{y}\in\mC, \mathbf{x}\neq\mathbf{y}}\{d_{\sf P}(\mathbf{x},\mathbf{y})\}.
\end{eqnarray*}
For $\mathbf{x},\mathbf{y}$ in $\F_q^n,$ let $0< d_{\sf H}(\mathbf{x},\mathbf{y})<n$ be the Hamming distance between $\mathbf{x}$ and $\mathbf{y}$. Then, we have
\begin{eqnarray}~\label{eq:wt}
d_{\sf H}(\mathbf{x},\mathbf{y})+1<d_{\sf P}(\mathbf{x},\mathbf{y})<2d_{\sf H}(\mathbf{x},\mathbf{y}).
\end{eqnarray}
In the extreme cases, where $d_{\sf H}(\mathbf{x},\mathbf{y})$ equals $0$ or $n,$ clearly $d_{\sf H}(\mathbf{x},\mathbf{y})=d_{\sf P}(\mathbf{x},\mathbf{y}).$

A code over $\F_q$ of length $n$ with size $M$ and minimum pair distance $d_{\sf P}$ is called an $(n,M,d_{\sf P})_q$-symbol-pair code. Similar to classical Hamming metric codes, we can define the rate and the relative pair distance of an $(n,M,d_{\sf P})_q$-symbol-pair code $\mC$ by
\begin{eqnarray*}
R(\mC)=\frac{\log_q|\mC|}{n}~~~~~~{and}~~~~~~\delta(\mC)=\frac{d_{\sf P}-2}{n},
\end{eqnarray*}
In literature, the relative distance of $\mC$ is defined by $\frac{d_{\sf P}}{n}.$ However, our definition of relative minimum distance given above will bring us advantage to handle some upper bounds like the Singleton bound.

The minimum pair distance is one of the important parameters for a symbol-pair code. A code $\mC$ with minimum pair distance $d_{\sf P}$ can uniquely correct $t$ pair errors if and only if $d_{\sf P}\geq 2t+1$ see~\cite{Y.C2011}. Hence, it is desirable to keep minimum pair distance $d_{\sf P}$ as large as possible for a symbol-pair code with fixed $n.$ It has been shown~\cite{YM2013} that an $(n,M,d_{\sf P})_q$-symbol-pair code $\mC$ must obey the following version of the Singleton bound.

\begin{lemma}(Singleton Bound) Let $q\geq 2$ and $2\leq d_{\sf P}\leq n.$ If $\mC$ is an $(n,M,d_{\sf P})_q$-symbol-pair code, then
\begin{eqnarray*}
M\leq q^{n-d_{\sf P}+2}.
\end{eqnarray*}
An alternative way to state the Singleton bound for a symbol-pair code $\mC$ in term of its rate and relative minimum pair distance is
\begin{eqnarray*}
R(\mC)+\delta(\mC)\leq 1.
\end{eqnarray*}
\end{lemma}

An $[n, k, d_{\sf P}]_q$ symbol-pair code is an $\F_q$-linear code over $\F_q$ of length $n$, dimension $k$ and minimum pair distance $d_{\sf P}.$

The symbol-pair ball, as an analog to the Hamming metric ball, is used to count the number of words within a given pair distance.
%%%%%%%%%%%%%%%%%%%%%%%%%%%%%%%%%
\begin{defn} (Symbol-pair Ball) For a word $\mathbf{y}\in\F_q^n$ and a nonnegative real number $r,$ the symbol-pair ball centered at $\mathbf{x}$ with radius $r$ is defined by
\begin{eqnarray*}
\mathcal{B}_{\sf P}(\mathbf{x},r)=\{\mathbf{y}\in\F_q^n: d_{\sf P}(\mathbf{x}, \mathbf{y})\leq r\}.
\end{eqnarray*}
\end{defn}
%%%%%%%%%%%%%%%%%%%%%%%%%%%%%%%%%

\begin{prop} (see in \cite{Y.C2011})
For any $\mathbf{x}\in\F_q^n$, the symbol-pair ball $\mathcal{B}_{\sf P}(\mathbf{x},d)$ has size
\begin{eqnarray}\label{eq:2.1}
|\mathcal{B}_{\sf P}(\mathbf{x},d)|=1+\sum_{i=1}^d\sum_{k=\left\lceil{\frac{i}{2}}\right\rceil}^{i-1}D(n,k,i-k)(q-1)^k,
\end{eqnarray}
where
\begin{eqnarray*}
\nonumber D(n,\ell,w)&=&{\binom{\ell-1}{w-1}}\left[{\binom{n-\ell-1}{w}}+2{\binom{n-\ell-1}{w-1}}\right]+{\binom{n-\ell-1}{w-1}}{\binom{\ell-1}{w}}\\
&=&\frac{n}{w}\cdot{\binom{\ell-1}{w-1}}{\binom{n-\ell-1}{w-1}}.
\end{eqnarray*}
\end{prop}
%%%%%%%%%%%%%%%%%%%%%%%%%%%%%%%%%
%Then, we introduce the Gilbert-Varshamov bound of the symbol-pari codes.
As in the Hamming metric, the codes in the symbol-pair metric also achieve the following Gilbert-Varshamov Bound.
%%%%%%%%%%%%%%%%%%%%%%%%%%%%%%%%%%%%%%%
\begin{lemma} (Asymptotic Gilbert-Varshamov Bound, see in \cite{CL})
There exists a family of $a$-ary $(n,M,d)$-symbol-pair codes with rate $R=\lim_{n\rightarrow\infty}\frac{\log_qM}n$ and relative pair distance $\delta=\lim_{n\rightarrow\infty}\frac d n$ satisfying
\begin{eqnarray*}
R\geq 1-\mathop{\max}\limits_{0\leq\frac{\theta}{2}\leq \beta\leq \theta\leq \delta}~\left(\beta H_q\left(\frac{2\beta-\theta}{\beta}\right)+(1-\beta)H_q\left(\frac{\theta-\beta}{1-\beta}\right)\right).
\end{eqnarray*}
\end{lemma}
%\begin{proof}
%The proof is a standard argument based on the asymptotic behaviour of the size of the symbol-pair ball.
%We omit it here.
%By Equation~\eqref{eq:2.2}, we have $|\mathcal{B}_{\sf P}(\mathbf{a}, d)|\leq n^2\cdot (\beta,\theta)_{\max}.$ To obtain asymptotic bounds on code rates it is useful to bound the ratio $q^n/|\mathcal{B}_{\sf P}(\mathbf{a}, d)|,$ such that
%\begin{eqnarray*}
%\frac{q^n}{|\mathcal{B}_{\sf P}(\mathbf{a}, d)|}\geq\frac{q^n}{(\beta,\theta)_{\max}}.
%\end{eqnarray*}
%Taking the logarithm and normalizing by $n,$ we get
%\begin{eqnarray*}
% \mathop{\lim}\limits_{n\to \infty}\frac{\log_q\frac{q^n}{|\mathcal{B}_{\sf P}(\mathbf{a}, d))|}}{n}\geq1-\mathop{\max}\limits_{0\leq\frac{\theta}{2}\leq \beta\leq \theta\leq \delta}~\left(\beta H_q\left(\frac{2\beta-\theta}{\beta}\right)+(1-\beta)H_q\left(\frac{\theta-\beta}{1-\beta}\right)\right).\end{eqnarray*}
 %\end{proof}

\begin{rmk}
 Figure~\ref{GVbound} reveals the gap between the Gilbert-Varshamov bound in symbol-pair metric and in Hamming metric when $q=17.$
 In other words, the codes attaining this bound in symbol-pair metric achieves better trade-off in terms of rate and relative distance.

\begin{figure}[H]
\centering
\includegraphics[scale=0.56]{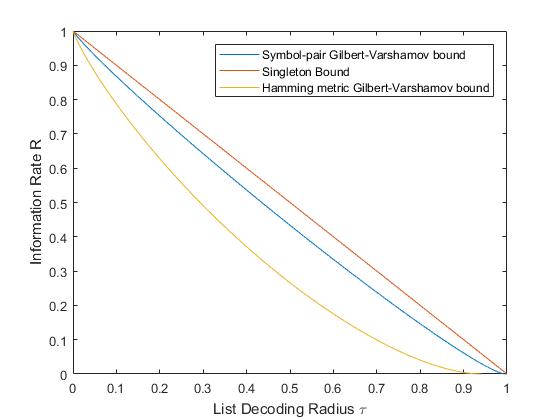}
\caption{Comparison of the Gilbert-Varshamov bound in Hamming Metric and Symbol-pair Metric.}\label{GVbound}
\centering
\end{figure}
\end{rmk}

We now proceed to the definition of list decoding of symbol-pair codes.
\begin{defn}
For a real $\tau\in(0,1),$ a symbol-pair code $\mC\subseteq\F_q^n$ is said to be $(\tau n,L)_{\sf P}$-list decodable, if for every $\mathbf{x}\in\F_q^n,$ we have
\begin{eqnarray*}
|\mathcal{B}_{\sf P}(\mathbf{x},\tau n)\cap\mC|\leq L.
\end{eqnarray*}
\end{defn}

\section{Bounds on the list decoding radius of Symbol-pair Codes}
\subsection{An upper bound on list decodbility of symbol-pair codes }

The Gilbet-Varshamov bound plays a role as an upper bound on the list decoding radius of codes under various metrics, i.e., the Hamming metric codes~\cite{Gur2005}, rank-metric codes~\cite{Ding2015} and cover-metric codes~\cite{Liu2018}.
It is not surprised that the Gilbert-Varshamov bound is also an upper bound on the list decoding radius of the symbol-pair codes.

%In this subsection, we show that any symbol-pair code cannot be list decoded beyond the Gilbert-Varshamov bound.
%We briefly state the idea behind our proof. By lower bounding the size of the symbol-pair ball, one can show that each codeword in this symbol-pair code has many nearby\footnote{Here, "nearby" means within the list decoding radius.} vectors. Averagely speaking, there must exist a vector that has many codewords nearby.

In this subsection, we show that list decoding of any symbol-pair code cannot exceed the Gilbert-Varshamov bound. The idea of our proof is based on counting the words in a symbol-pair ball. We firstly estimate the size of a symbol-pair ball.
\begin{lemma}~\label{lemma:limit}
Given a vector $\mathbf{a}\in\F_q^n$, the size of the symbol-pair ball ${\mathcal{B}_{\sf P}(\mathbf{a}, \delta n)}$ satisfies
\begin{eqnarray}~\label{eq:bounds 1}
|\mathcal{B}_{\sf P}(\mathbf{a}, \delta n)|=q^{\kappa_{sp}(\delta)n+o(n)},
\end{eqnarray}
where
\begin{eqnarray}~\label{eq:ka}
\kappa_{sp}(\delta)=\max_{{0\leq\frac{\theta}{2}\leq \beta\leq \theta\leq \delta}}\beta H_q\left(\frac{2\beta-\theta}{\beta}\right)+(1-\beta )H_q\left(\frac{\theta-\beta}{1-\beta}\right),\end{eqnarray}
and $H_q(x)=x\log_q(q-1)-x\log_qx-(1-x)\log_q(1-x)$ is the $q$-ary entropy function.
\begin{proof}
By the equation~\eqref{eq:2.1}, the size of the symbol-pair ball is
\begin{eqnarray*}
|\mathcal{B}_{\sf P}(\mathbf{a}, \delta n)|&=&1+\sum_{i=1}^{\delta n}\sum_{k=\lceil{\frac{i}{2}}\rceil}^{i-1}\frac{n}{i-k}\cdot{\binom{k-1}{i-k-1}}{\binom{n-k-1}{i-k-1}}(q-1)^k.
\end{eqnarray*}
Let $k=\beta n$ and $i=\theta n,$ for some reals $\beta\in(0,1)$ and $\theta\in(0,1),$ we have
\begin{eqnarray*}
\binom{k-1}{i-k-1}=2^{\beta nH_2(\frac{2\beta-\theta}{\beta})+o(n)}, ~\binom{n-k-1}{i-k-1}=2^{(1-\beta)nH_2(\frac{\theta-\beta}{1-\beta})+o(n)},
\end{eqnarray*}
this implies
\begin{eqnarray*}
\frac{n}{i-k}\cdot{\binom{k-1}{i-k-1}}{\binom{n-k-1}{i-k-1}}(q-1)^k=q^{\beta n H_q\left(\frac{2\beta-\theta}{\beta}\right)+(1-\beta )nH_q\left(\frac{\theta-\beta}{1-\beta}\right)+o(n)}.
\end{eqnarray*}

Thus \[q^{\kappa_{sp}(\delta)n+o(n)}\le |\mathcal{B}_{\sf P}(\mathbf{a}, \delta n)|\le (\delta n)^2q^{\kappa_{sp}(\delta)n+o(n)}=q^{\kappa_{sp}(\delta)n+o(n)}.\]
The desired result follows.
\end{proof}
\end{lemma}

To simplify the notation, we denote $\kappa_{sp}(\delta)$ by $\kappa_{sp}$ if there is no confusion.

\begin{rmk}
 Lemma~\ref{lemma:limit} simply says that
\begin{eqnarray*}
\mathop{\lim}\limits_{n\to \infty}\frac{\log_q|\mathcal{B}_{\sf P}(\mathbf{a}, \delta n))|}{n}=\kappa_{sp}.
\end{eqnarray*}
\end{rmk}

The following theorem shows that the Gilbert-Varshamov bound is an upper bound on the list decoding radius of symbol-pair codes.

\begin{theorem}~\label{Thm:1}
Assume that a symbol-pair code $\mC$ of rate $R$ is $(\tau n, L)_{\sf P}$-list decodable with list size $L={\poly}(n)$. Then, the rate $R$ of $\mC$ must obey
\begin{eqnarray*}
R\leq 1-\kappa_{sp}(\tau)=1-\mathop{\max}\limits_{0\leq\frac{\theta}{2}\leq \beta\leq \tau}~\left(\beta H_q\left(\frac{2\beta-\theta}{\beta}\right)+(1-\beta)H_q\left(\frac{\theta-\beta}{1-\beta}\right)\right)
\end{eqnarray*}
for all sufficiently large $n$, where $\kappa_{sp}(\tau)$ is given in \eqref{eq:ka}.
\begin{proof}
We prove it by contradiction. Simply denote $\kappa_{sp}(\tau)$ by $\kappa_{sp}$. Assume that there exists a symbol-pair code $\mC$ of rate $R$ such that
$R\geq 1-\kappa_{sp}+\epsilon$ for some positive constant $\epsilon$.
Let $L$ be the upper bound of the list size of this code. Define the set
\begin{eqnarray*}
\mA=\{(\bc,\bv): d_{\sf P}(\bc,\bv)\leq \tau n, \bc\in \mC, \bv \in \F_q^n\}.
\end{eqnarray*}
We find two ways to calculate the size of this set.
First, for every vector $\bv$ in $\F_q^n$, it holds that
$|B_{\sf P}(\bv,\tau n)\cap \mC|\leq L$. This implies
\begin{eqnarray*}
|\mA|=\sum_{\bv\in \F_q^n}|B_{\sf P}(\bv,\tau n)\cap \mC|\leq q^nL.
\end{eqnarray*}
On the other hand, by Lemma \ref{lemma:limit} we have $|B_{\sf P}(\bc,\tau n)|\geq q^{\kappa_{sp}n-\frac{\epsilon}2 n}$ for all sufficiently large $n$. Thus

\begin{eqnarray*}
|\mA|=\sum_{\bc\in \mC}|B_{\sf P}(\bc,\tau n)|\geq q^{Rn} q^{\kappa_{sp}n-\frac{\epsilon}2n }.
\end{eqnarray*}
Combining them together gives us
\begin{eqnarray*}
L\geq q^{Rn+\kappa_{sp}n-\frac{\epsilon}2n -n}\geq  q^{\frac{\epsilon}{2}n}.
\end{eqnarray*}
A contradiction occurs.
\end{proof}
\end{theorem}

%\begin{rmk}
%An alternative proof of Theorem~\ref{Thm:1}. Based on~Equation~\eqref{eq:wt}, we know $\wt_{\sf H}(\mathbf{a}-\mathbf{b})\leq \wt_{\sf P}(\mathbf{a}-\mathbf{b})$ for any word $\mathbf{a}, \mathbf{b}\in\F_q^{n}.$ Since any word that is in $\mB_{\sf H}(\mathbf{a}, \tau n)$ is also in $\mB_{\sf P}(\mathbf{a}, \tau n)$ (but not vice versa), so $|\mB_{\sf P}(\mC, \tau n)\cap \mC|\leq |\mB_{\sf H}(\mC, \tau n)\cap \mC|$ with the same center and the list decoding radius. In other words, a $(\tau n, \mL)_{\sf H}$-list decodable code $\mC$ of $\F_q^n$ with respect to the Hamming metric code is $(\tau n,\mL)^{\sf P}$-list decodable with respect to the symbol-pair code. The Gilbert-Varshamov bound was known as an upper bound on the list decoding radius of codes under Hamming metric codes in~\cite{Gur2005}. Hence, the limit of list decodability of symbol-pair codes still is the Gilbert-Varshamov bound.
%
%Note that when the field size $q=\omega(1),$ the limit of the list decodability of symbol-pair codes is the Singleton bound.
%\end{rmk}

%In next section, we will study a random symbol-pair code can be list decoded up to the Gilbert-Varshamov bound. When the field size $q$ is sufficiently large, symbol-pair codes can be list decoded up to the Singleton bound.
\subsection{List decoding of random symbol-pair codes}
In the previous subsection, we show that list decodability of every symbol-pair codes does not exceed the Gilbert-Varshamov bound.
In this subsection, we investigate list decodability of random symbol-pair codes. We show that random symbol-pair codes can be list decoded up to the Gilbert-Varshamov bound with high probability. In particular, most symbol-pair codes can be list decoded up to the Gilbert-Varshamov bound with constant list size $O(1/\epsilon)$, %The same result holds for linear symbol-pair codes but the list size grows up to $\exp(O(1/\epsilon))$.
\begin{theorem}\label{thm:random}
For small $\epsilon\in(0,1)$ with a probability at least $1-q^{-n}$, a random symbol-pair code $\mC\subseteq\F_q^n$ of rate
\begin{eqnarray*}
R= 1-\kappa_{sp}(\tau)-\epsilon=1-\mathop{\max}\limits_{0\leq\frac{\theta}{2}\leq \beta\leq \tau}~\left(\beta H_q\left(\frac{2\beta-\theta}{\beta}\right)+(1-\beta)H_q\left(\frac{\theta-\beta}{1-\beta}\right)\right)-\epsilon
\end{eqnarray*}
is $(\tau n, O(1/\epsilon))_{\sf P}$-list decodable for sufficiently large $n.$
\end{theorem}
\begin{proof}
 Put $L=\left\lceil\frac{4}{\epsilon}\right\rceil-1.$ By Lemma~\ref{lemma:limit}, for all sufficiently large $n,$ we have $|\mathcal{B}_{\sf P}(\mathbf{a},\tau n)|\leq q^{\kappa_{sp} n+\frac{\epsilon}{2}n}.$  Pick a symbol-pair code $\mC$ with size $q^{Rn}$ uniformly at random. Let us upper bound the probability that $\mC$ is not $(\tau n, L)_{\sf P}$-list decodable.

If $\mC$ is not $(\tau n, L)_{\sf P}$-list decodable, there exists a word $\mathbf{a}\in \mathbb{F}_{q}^{n}$ and a subset $\mS\subseteq \mC$ with $|\mS|=L+1$ such that $\mS\subseteq \mB_{{\sf P}}(\mathbf{a},\tau n)$.
The probability that codeword $\mathbf{c}\in \mC$ is contained in $\mB_{{\sf P}}(\mathbf{a},\tau n)$ is
 \begin{equation}\label{eq:2.3}
 \Pr[c\in \mB_{{\sf P}}(\mathbf{a},  \tau n)]=\frac{|\mB_{{\sf P}}(\mathbf{a}, \tau n)|}{q^{n}}
\leq  q^{\kappa_{sp}n+\frac{\epsilon}{2}n}\cdot q^{-n}.
 \end{equation}
 Let $E_{\mathbf{a},\mS}$ be the event that all codewords in $\mS$ are contained in $\mB_{{\sf P}}(\mathbf{a}, \tau n)$.
By Equation~\eqref{eq:2.3}, we have
\begin{equation*}
\Pr[E_{\mathbf{a},S}]\leq\left(\frac{|\mB_{{\sf P}}(\mathbf{a},  \tau n)|}{q^{n}}\right)^{L+1}\leq {\left( q^{\kappa_{sp}n+\frac{\epsilon}{2}n}\cdot q^{-n}\right)}^{L+1}.
\end{equation*}
Taking the union bound over all $q^{n}$ choices of $\mathbf{a}$ and $\mS$ over any $(L+1)$-subsets of $\mC$, we have
\begin{eqnarray*}
\nonumber\sum_{\mathbf{a},\mS}\Pr[E_{\mathbf{a},\mS}]&\leq& q^{n}\cdot\binom{|\mC|}{L+1}\cdot{\left( q^{\kappa_{sp}n+\frac{\epsilon}{2}n}\cdot q^{-n}\right)}^{L+1}\\
\nonumber&\leq&q^{n}\cdot|\mC|^{L+1}\cdot q^{(\kappa_{sp}n+\frac{\epsilon}{2}n)(L+1)}\cdot q^{-n(L+1)}\\
\nonumber&\leq&q^{n}\cdot q^{Rn(L+1)}\cdot  q^{(\kappa_{sp}n+\frac{\epsilon}{2}n-n)(L+1)} \\
\nonumber&=&q^{n(L+1)\left(\frac{1}{L+1}+R+\kappa_{sp}+\frac{\epsilon}{2}-1\right)}\\
\nonumber&\leq&q^{n(L+1)\left(\frac{\epsilon}{4}+R+\kappa_{sp}+\frac{\epsilon}{2}-1\right)}\nonumber\leq q^{-n}.\\
\end{eqnarray*}
The last inequality holds since $R= 1-\kappa_{sp}-\epsilon.$ Thus, a symbol-pair code $\mC$ with rate $R$ is not $(\tau n,L)_{\sf P}$-list decodable with probability at most $q^{-n}$.
\end{proof}

\subsection{The Johnson-type bound}
The Johnson-type bound in the topic of list decoding usually provides a lower bound on list decoding radius in terms of minimum distance of a code. However, for some metrics such as rank-metric, the Johnson-type bound does not exist. In this section, we show that one has a  Johnson-type bound for pair metric. On the hand hand, there is an evidence showing that the Johnson-type bound given in this subsection is tight.

%By the definition of the Johnson-type bound, our goal is to prove a lower bound on the list decoding radius of the symbol-pair code, i.e., any symbol-pair code can be list decodable up to this bound.
\begin{theorem}(Johnson-type Bound)\label{thm:JT}
Any symbol-pair code $\mC$ in $\F_q^n$ with  relative distance $\delta$ is $(\tau n, 2(q^2-1)nd)$-list decodable for
$$\tau=\frac{q^2-1}{q^2}\left(1-\sqrt{1-\frac{q^2\delta}{q^2-1}}\right)$$
\end{theorem}
\begin{proof}
We fix a vector $\by\in \F_q^n$. Assume that $B_{\sf P}(\by,\tau n)\cap \mC=\{\bc_1,\ldots,\bc_L\}$ for $L$. Our goal is to bound the size $L$. Let $\bv_i=\bc_i-\by$. Since $d_{\sf P}\geq d_{\sf P}(\mC)$, we have $d_{\sf P}(\bv_i,\bv_j)=d_{\sf P}(\bc_i,\bc_j)\geq \delta n$
for every pair $(i,j)$ and $\wt_{\sf P}(\bv_i)\leq \tau n$ for every $i$. We denote $\bv_i$ as $(v_{i,1},\ldots,v_{i,n})\in \F_q^n$.
By the definition of symbol-pair error, we have
\begin{eqnarray*}
\frac{L(L-1)\delta n}{2}&\leq&\sum_{1\leq i<j\leq L}d_{\sf P}(\bv_i,\bv_j)=\sum_{1\leq i<j\leq L}|\{k: (v_{i,k},v_{i,k+1})
\neq (v_{j,k},v_{j,k+1})\}|\\&=&\sum_{k=1}^n|\{(i,j):(v_{i,k},v_{i,k+1})\neq (v_{j,k},v_{j,k+1}), 1\leq i<j\leq L\}|.
\end{eqnarray*}
Next, we fix the coordinate pair $(1,2)$. Let $x_{a,b}$ be the number of pairs $(a,b)$ among the set $\{(v_{i,1},v_{i,2})\in \F_q^2:1\leq i\leq L\}$. It is clear that $\sum_{(a,b)\in \F_q^2}x_{a,b}=L$. It follows that
\begin{eqnarray*}
&&|\{(i,j):(v_{i,1},v_{i,2})\neq (v_{j,1},v_{j,2}), 1\leq i<j\leq L\}|=
\sum_{(a,b)\in \F_q^2}x_{a,b}(L-x_{a,b})\\&=&\left(L^2-x^2_{0,0}-\sum_{(a,b)\in \F_q^2/(0,0)}x^2_{a,b}\right)
\leq \left(L^2-x^2_{0,0}-\frac{1}{q^2-1}\left(\sum_{(a,b)\in \F_q^2/(0,0)}x_{a,b}\right)\right)\\&=&\left(L^2-x^2_{0,0}-\frac{1}{q^2-1}(L-x_{0,0})^2\right)
\end{eqnarray*}
The inequality above is due to the Cauthy-Schwarz inequality.
We can apply this argument to every pair of adjacent coordinates $(k,k+1)$.
Let $a_k$ be the number of pairs $(0,0)$ among the set $\{(v_{i,k},v_{i,k+1})\in \F_q^2:1\leq i\leq L\}$.
Putting these two formulas together gives us
\begin{eqnarray*}
\frac{L(L-1)\delta n}{2}&\leq& nL^2-\sum_{k=1}^n \left(a^2_k+\frac{1}{q^2-1}(L-a_k)^2\right)\\
&=& \frac{2}{q^2-1}L\sum_{k=1}^n a_k-\frac{q^2}{q^2-1}\sum_{k=1}^n a^2_k+n\frac{q^2-2}{q^2-1}L^2\\&\leq&
-\frac{q^2}{n(q^2-1)}\left(\sum_{k=1}^n a_k\right)^2+\frac{2}{q^2-1}L\left(\sum_{k=1}^n a_k\right)+n\frac{q^2-2}{q^2-1}L^2
\end{eqnarray*}
Let $\sum_{k=1}^n a_k=Le$ and we then have
$$
-\frac{q^2}{n(q^2-1)}L^2e^2+\frac{2}{q^2-1}L^2e-\frac{L(L-1)\delta n}{2}+n\frac{q^2-2}{q^2-1}L^2\geq 0.
$$
This implies
\begin{equation}\label{eq:listsize}
L\leq \frac{2\delta n}{\frac{q^2e^2}{n(q^2-1)}-\frac{2e}{q^2-1}+\delta n-\frac{(q^2-2)n}{q^2-1}}.
\end{equation}
The condition $\frac{q^2e^2}{n(q^2-1)}-\frac{2e}{q^2-1}+\delta n-\frac{(q^2-2)n}{q^2-1}>0$ leads to
$$
\frac{e}{n}<\frac{1}{q^2}+\frac{q^2-1}{q^2}\sqrt{1-\frac{q^2\delta}{q^2-1}}.
$$
This implies
$$(n-q^2e)>(q^2-1)n\sqrt{1-\frac{q^2\delta}{q^2-1}}.$$
Squaring both sides and observing that $\delta=\frac{d}{n}$ yields
$$
(n-q^2e)^2>(q^2-1)^2n^2-(q^2-1)q^2nd.
$$
Since both sides are integers, we obtain $(n-q^2e)^2\geq(q^2-1)^2n^2-(q^2-1)q^2nd+1.$
Observe that $\eqref{eq:listsize}$ is equivalent to
$$
L\leq \frac{ 2(q^2-1)dn}{(n-q^2e)^2-(q^2-1)^2n^2+(q^2-1)q^2nd}\leq 2(q^2-1)dn.
$$
Then, the desired result follows.
\end{proof}
%\begin{rmk}
%This Johnson-type bound has exactly the same form as the Johnson bound for Hamming metric codes over $\F_{q^2}$.
%\end{rmk}

%In Hamming distance, there are many works dedicated to designing codes that are list decodable beyond the Johnson bound.
%Most of them come from the subcodes of the Reed-Solomon code. One might conjecture that the same construction might yield the symbol-pair code with poor list decoding radius as well. However, \eqref{eq:wt} says that the %symbol-pair distance could be twice as the Hamming distance.
%This implies that the code that has poor list decoding radius in Hamming metric is not necessarily a symbol-pair code with poor list decoding radius. Thus, we must learn the structure of these subcodes and hopefully transform them to symbol-pair codes that still have relatively poor list decoding radius.
One may wonder if the Johnson-type Bound derived in this subsection is optimal.
We find that the codes in~\cite{BKR} can be used to illustrate that the Johnson-type bound derived in this subsection is at least very close to optimality though we do not have an affirmative answer.

The paper \cite{BKR} focused on the low-degree linearized polynomials that agrees with a given high-degree linearized polynomials on many coordinates. The following lemma summarize their results. Fix $n$ distinct elements $\Ga_1,\dots,\Ga_n$. For a polynomial $f(x)\in\F_q[x]$, we denote by $\bc_f$ the vector $(f(\Ga_1),\dots,f(\Ga_n))$.  We abuse notations and denote by $d_{\sf P}(a(x),b(x))$ (and $d_{\sf H}(a(x),b(x))$, respectively) the symbol-pair distance (and the Hamming distance, respectively) between  $\bc_a$ and $\bc_b$.
% i.e.,
%\begin{eqnarray*}
%&&d_{\sf H}(a(x),b(x))=d_{\sf H}((a(\alpha_1),\ldots,a(\alpha_n),(b(\alpha_1),\ldots,b(\alpha_n)),\\
%&&d_{\sf P}(a(x),b(x))=d_{\sf P}((a(\alpha_1),\ldots,a(\alpha_n),(b(\alpha_1),\ldots,b(\alpha_n)),
%\end{eqnarray*}
%where $a(\alpha_1),\ldots,a(\alpha_n)$ and $b(\alpha_1),\ldots,b(\alpha_n)$ are the evaluations of $a(x)$ and $b(x)$, respectively.
\begin{lemma}[{\cite[Theorem 2.1]{BKR}}]\label{thm:ham}
Let $\ell$ be a prime power and $m$ a positive integer. Put $q=\ell^m$. Let $u$ and $v$ be integers such that $0\leq u\leq v\leq m$. Then,
there is a family $\mP\subseteq \F_{\ell^m}[X]$ of \emph{linearized}\footnote{They did not mention ''linearized'' in this theorem. Judged from their construction, $\mP$ is indeed a family of linearized polynomials.} polynomials of degree $\ell^u$ and a linearized polynomial $w(x)$ such that
\begin{enumerate}
\item $|\mP|\geq \ell^{(u+1)m-v^2}$;
\item for all $P(x)\in \mP$, $d_{\sf H}(P(x),w(x))\leq \ell^m-\ell^v$;
\item $w(x)=x^{\ell^v}+\sum_{i=u+1}^{v-1}a_ix^{\ell^i}$.
\end{enumerate}
\end{lemma}
Based on this lemma, we have the following result that leads to some symbol-pair codes we need to illustrate optimality of the Johnson-type Bound given in this subsection.
\begin{lemma}\label{thm:listdecoding}
Let ${\ell}$ be a prime power and $m$ a positive integer. Put $q=\ell^m$.  Let $u$ and $v$ be integers such that $0\leq u\leq v\leq m$. Then,
there is a family $\mP\subseteq \F_{{\ell}^m}[X]$ of linearized polynomials of degree ${\ell}^u$ and a linearized polynomial $w(x)$ such that
\begin{enumerate}
\item $|\mP|\geq {\ell}^{(u+1)m-v^2}$;
\item for all $P(x)\in \mP$, $d_{\sf P}(P(x),w(x))\leq {\ell}^m-\frac{({\ell}-2)}{{\ell}-1}({\ell}^v-1)$;
\item $w(x)=x^{{\ell}^v}+\sum_{i=u+1}^{v-1}a_ix^{{\ell}^i}$.
\end{enumerate}
\end{lemma}
\begin{proof}
Let $[v]=\{\lambda v:\lambda \in \F^*_{\ell}\}$ and $h=\frac{{\ell}^m-1}{{\ell}-1}$. As we know, the set $\F^*_{{\ell}^m}$ can be partitioned into $h$ disjoint subsets $[v_1],\ldots,[v_h]$.
Since the distance of symbol-pair code is greatly affected by the order of its coordinates,
we start our proof by arranging the order of coordinates. Given a polynomial $f(x)\in \F_{{\ell}^m}[X]$, the codeword generated by $f(x)$ is $(f(0), f([v_1]), f([v_2]),\ldots,f([v_h]))$ where
$f([v_i])\triangleq(f(\lambda v_i))_{\lambda \in \F^*_{\ell}}$.
Let $\mP$ and $w(x)$ be the family of linearized polynomials and linearized polynomials given by Theorem~\ref{thm:ham}.
For any $P(x)\in \mP$, let us bound the symbol-pair distance of $P(x)$ and $w(x)$ under the above order of coordinates.
By Theorem~\ref{thm:ham}, the linearized polynomial $g_P(x)\triangleq P(x)-w(x)$ has at least ${\ell}^v$ roots. Moreover, if $u\in \F^*_{{\ell}^m}$ subject to
$g_P(u)=0$, then $g_P(\lambda u)=0$ for every $\lambda\in \F^*_{{\ell}^m}$. Assume that $u\in [v_i]$ and we have $g_P([v_i])=\mathbf{0}\in \F_{{\ell}^m}^{{\ell}-1}$. It follows that $g_P([v_i])$ contributes ${\ell}-2$ pairs of symbols $(0,0)\in \F_{{\ell}^m}^2$. In summary, the ${\ell}^v$ roots of $g_P(x)$ yields at least $({\ell}-2)\frac{({\ell}^v-1)}{{\ell}-1}$ pairs of adjacent coordinates whose symbol patterns are $(0,0)\in \F_{{\ell}^m}^2$.
The desired result follows since
$$d_{\sp P}(P(x),w(x))=\wt_{\sp P}(P(x)-w(x))\leq {\ell}^m-\frac{({\ell}-2)}{{\ell}-1}({\ell}^v-1).$$
\end{proof}
\begin{ex}{\rm In this example, we  illustrate optimality of the Johnson-type bound given in this subsection.

 We follows the parameter setting in \cite{BKR}. Let ${\ell}$ be a prime power and $m$ a positive integer. Put $q=\ell^m$.
Lemma~\ref{thm:listdecoding} yields a symbol-pair code with list decoding radius at most $1-\frac{{\ell}-2}{{\ell}-1}{\ell}^{v-m}$. The dimension of this code is $K:={\ell}^u$ and the length of this code is $N:={\ell}^m$. Setting $u=\delta m$ and $v=\rho m$ gives the list size
$|\mP|\geq N^{(\delta-\rho^2)\log_{\ell} N}$ which is super-polynomial in length $N$ for any constant $\delta>\rho^2$. To compare it with our Johnson-type bound, we set $\delta=1-\gamma$ and $\rho=1-\frac{\gamma}{2}-\frac{\gamma^2}{4}$ for small constant $\gamma$.
One can check that it satisfies $\delta>\rho^2$ for small constant $\gamma$. Let ${\ell}=\frac{1}{\epsilon}$ and the relative decoding radius then becomes
$$1-\frac{{\ell}-2}{{\ell}-1}{\ell}^{\rho m-m}=1-(1-\epsilon)N^{-\frac{\gamma}{2}-\frac{\gamma^2}{4}}.$$
On the other hand, our Johnson-type bound gives the relative list decoding radius $(1-\frac{1}{N^2})(1-N^{-\frac{\gamma}{2}})\approx 1-N^{-\frac{\gamma}{2}}$. Thus, the upper bound is very close to the Johnson-type bound for rate $R=N^{-\gamma}$. This implies that the Johnson-type bound given in this subsection is very close to optimality if it is not optimal.

%$$1-\frac{{\ell}-2}{{\ell}-1}{\ell}^{v-m}=1-\frac{{\ell}-2}{{\ell}-1}{\ell}^{v-m}\leq 1-{\ell}^{(1-\epsilon)((u-m)\frac{\sqrt{m}}{\sqrt{m}+\sqrt{u}})}\leq
%1-(q^{u-m})^{(1-\epsilon)\frac{\sqrt{m}}{\sqrt{m}+\sqrt{u}}}1$$
}\end{ex}
%%%%%%%%%%%%%%%%%%%%%%%%%%%%%%%%%%

%\section{Explicit Constructions}
\section{List decoding of Reed-Solomon codes beyond the Johnson-type bound}
It is well known that  any Reed-Solomon codes can be efficiently list decoded up to the Johnson bound for the Hamming metric with the help of famous Guruswami-Sudan list decoding algorithm.
On the other hand, some evidence shows that there exist Reed-Solomon codes and subcodes of Reed-Solomon codes that can not be list decoded slightly beyond the Johnson bound for the Hamming metric. Given the importance of Reed-Solomon code in both theory and practice, one would like to clearly understand the limits to the list decoding issue of Reed-Solomon codes. However, we are still far away from this goal anyway for the Hamming metric. It is not even clear whether there exist Reed-Solomon codes that can be list decoded beyond the Johnson bound for the Hamming metric.

On the other hand, one also wonders if  Reed-Solomon codes can be list decoded  beyond the Johnson bound for the pair metric.
 In this subsection, we give this question an affirmative answer by showing that  Reed-Solomon codes can indeed be list decoded  beyond the Johnson-type bound.

 The construction comes from the folded Reed-Solomon code. Let us first explain the intuition behind this construction.
By the definition of symbol-pair error, each error corresponds to a pair of adjacent coordinates. In our list decoding algorithm,
instead of inputting the evaluations index by index, we input the evaluations pair by pair. The question arises whether we can exploit this input to improve our list decoding algorithm. Note that the famous Guruswami-Sudan list decoding algorithm fails to serve our purpose. We turn to the list decoding algorithm of folded Reed-Solomon code in~\cite{V.G2008} instead. Let $\gamma$ be a primitive element of $\F_q$.

We now consider list decoding of folded Reed-Solomon code. Let $\gamma$ be a primitive element of $\F_q$. Let $1\le k\le n\le q-1$. We encodes the polynomial $f$ of degree at most $k-1$ to the codewords
 $\bc_f:=(f(1),f(\gamma),\ldots,f(\gamma^{n-1}))$
 and
 \begin{equation}\label{eq:FC}
\bc_f^{(2)}:=\left(
  \begin{array}{ccccc}
    f(1) & f(\gamma) & f(\gamma^2)   &  \cdots   &  f(\gamma^{n-2}) )\\
    f(\gamma) & f(\gamma^2) & f(\gamma^3)  & \cdots &  f(\gamma^{n-1}) \\
  \end{array}
\right)
\end{equation}
Then the Reed-Solomon code $RS[n,k]$ and the folded Reed-Solomon $FRS[n,k]$ are defined by
 \begin{equation}\label{eq:RS}
RS[n,k]:=\{\bc_f:\; f\in\F_q[x],\; \deg(f)\le k-1\}.
\end{equation}
and
 \begin{equation}\label{eq:FRS}
FRS[n-1,k]:=\{\bc_f^{(2)}:\; f\in\F_q[x],\; \deg(f)\le k-1\}.
\end{equation}
respectively.
List decoding of folded Reed-Solomon codes were first considered in \cite{V.G2008}. The main idea of the following result can be found in \cite{V.G2008}. However, for the sake of completeness, let us derive an explicit list decoding algorithm of folded Reed-Solomon codes defined above.
\begin{lemma}\label{lem:4.1} The folded Reed-Solomon code $FRS[n-1,k]$ defined in \eqref{eq:FRS} is $(\tau (n-1),q)_{\sf H}$-list decodable with $\tau=\frac{2}3\times\frac{n-2-k}{n-1}$.
\end{lemma}
\begin{proof} Assume that  $\bc_f^{(2)}$ was transmitted and
\[\bb^{(2)}:=\left(
  \begin{array}{ccccc}
    a_1 & a_2 & a_3   &  \cdots   &  a_n \\
  b_1 & b_2 & b_3  & \cdots &  b_n \\
  \end{array}
\right)\]
is received with at most $\tau n$ errors. Thus, $d_{\sf H}( \bc_f^{(2)},\bb^{(2)})\le \tau (n-1)$.
Put $m=\lceil (n-k)/3\rceil$. Then one has $3m+k+2>n-1$. Consider the interpolation polynomial $Q(x,y_1,y_2):=a_0(x)+a_1(x)y_1+a_2(x)y_2\in \F_q[x,y_1,y_2]$ with coefficients of $a_i(x)$ to be determined subject to $\deg(a_0)\le m+k-1$, $\deg(a_1)\le m$ and  $\deg(a_2)\le m$.
Consider the homogenous equation system $a_0(\Gg^{i-1})+ a_1(\Gg^{i-1})a_i+a_2(\Gg^{i-1})b_i=0$ for $i=1,2,\dots,n-1$. For this equation system, coefficients of $a_i(x)$ are viewed as variables. Thus, there are $3m+k+2$ variables and $n-1$ equations. Hence, there are polynomials $a_0(x), a_1(x), a_2(x)\in\F_q[x]$ with $\deg(a_0)\le m+k-1$, $\deg(a_1)\le m$ and  $\deg(a_2)\le m$ that are not all zero such that  $a_0(\Gg^{i-1})+ a_1(\Gg^{i-1})a_i+a_2(\Gg^{i-1})b_i=0$ for $i=1,2,\dots,n-1$. Since $d_{\sf H}(\bc_f^{(2)},\bb^{(2)})\le \tau n$, there are at least $n-1-\tau (n-1)$ $i$'s such that $a_0(\Gg^{i-1})+ a_1(\Gg^{i-1})f(\Gg^{i-1})+a_2(\Gg^{i-1})f(\Gg^{i})=0$. Hence, the polynomial $a_0(x)+a_1(x)f(x)+a_2(x)f(\Gg x)$ has at least $n-1-\tau(n-1)$ roots. On the other hand, $\deg(a_0(x)+a_1(x)f(x)+a_2(x)f(\Ga x))\le m+k-1$ and we also have $n-1-\tau(n-1)>m+k-1$, this forces that $a_0(x)+a_1(x)f(x)+a_2(x)f(\Gg x)$ is identical to $0$. Note that $x^{q-1}-\Gg$ is irreducible and $x^q\equiv\Gg x\pmod{x^{q-1}-\Gg}$. This gives
\[0=a_0(x)+a_1(x)f(x)+a_2(x)f(\Gg x)\equiv a_0(x)+a_1(x)f(x)+a_2(x)f^q(x)\pmod{x^{q-1}-\Gg}.\]
In other words, $f(x)$ is a solution of the equation $a_0(x)+a_1(x)z+a_2(x)z^q=0$ over the field $\F_q[x]/(x^{q-1}-\Ga)\simeq\F_{q^{q-1}}$. Hence, this equation has at most $q$ roots in $\F_q[x]/(x^{q-1}-\Gg)$.  Since $\deg(f(x))<q-1$, the equation $a_0(x)+a_1(x)f(x)+a_2(x)f(\Gg x)=0$ has at most $q$ roots in $\F_q[x]$.
\end{proof}
By applying Lemma \ref{lem:4.1} and considering the relation between Hamming distance and pair distance, we immediately obtain the following result.
\begin{theorem}\label{thm:4.2}
The Reed-Solomon code $RS[n,k]$ over $\F_q$ for any $1\le k\le n\le q$  is $(\tau n,q)_{\sf P}$-list decodable with $\tau=\frac{2}3\times\frac{n-2-k}{n}$.
\end{theorem}

\begin{lemma}\label{lem:4.3}
The Reed-Solomon code $RS[n,k]$ over $\F_q$ for any $1\le k< n\le q$  has pair minimum distance at $n-k+2$.
\end{lemma}
\begin{proof}
Consider the polynomial $f(x)=\prod_{i=0}^{k-2}(x-\Gg^{i})$. Then the codeword $\bc_f$ has Hamming weight $n-k+1$ and the pair weight $n-k+2$. This completes the proof.
\end{proof}

\begin{theorem}\label{thm:4.3}
The Reed-Solomon code $RS[n,k]$ over $\F_q$ for any $1+n/2\le k< n\le q$  is $(\tau n,q)_{\sf P}$-list decodable with $\tau=\frac{2}3\Gd+o(1)$, where $\Gd=\frac{n-k+2}{n}$ is the relative pair minimum distance of  $RS[n,k]$. Hence, if $n$ is proportional $q$ and  $0<\Gd<\frac 34$, then $RS[n,k]$ can be list decoded beyond the  Johnson-type bound  with list size $O(n)$.
\end{theorem}
\begin{proof} When $n$ is proportional to $q$, the list size given in Theorem \ref{thm:4.2} is $O(n)$. For sufficiently large $n$ (thus $q$ is also large), the Johnson-type bound given in Theorem \ref{thm:JT} becomes $1-\sqrt{1-\delta}+o(1)$.
On the other hand, by Lemma \ref{lem:4.3}, the relative  minimum distance of  $RS[n,k]$ is $\Gd=\frac{n-k+2}{n}$ for $\Gd<1/2$.
Furthermore, it is easy to verify that $\frac{2}3\Gd>1-\sqrt{1-\delta}$ for $0<\delta<\frac{3}{4}.$
%A sharp contrast to the Reed-Solomon code in Hamming metric code, some Reed-Solomon code is list decodable beyond the Johnson-type bound in symbol-pair error.
\end{proof}


\begin{thebibliography}{99}

\bibitem{BKR}
E. Ben-Sasson, S. Kopparty and J. Radhakrishnan, Subspace Polynomial and Limits to List Decoding of Reed-Solomon Codes,
{\it IEEE Transactions on Information Theory}, vol. 56, no. 1, pp 113-120, 2010.

\bibitem{Y.C2011}
Y. Cassuto, M. Blaum, Codes for Symbol-Pair Read Channels,  {\it IEEE Transactions on Information Theory}, vol. 57, no. 12, pp 8011-8020, 2011.

\bibitem{CL}
Y. Cassuto and S. Litsyn, Symbol-pair codes: Algebraic constructions and asymptotic bounds, {\it IEEE International Symposium on Information Theory}, pp. 2348-2352 (2011).

\bibitem{YM2013}
Y. M. Chee, L. Ji, H. M. Kiah, C. Wang, J. Yin, Maximum distance separable codes for symbol-pair read channels, {\it IEEE Transactions on Information Theory}, vol. 59, no. 11, pp 7259-7267, 2013.

\bibitem{BGJ}
B. Ding, G. Ge, J. Zhang, T. Zhang and Y. Zhang, New constructions of MDS symbol-pair codes, {\it Des. Codes Cryptogr}. (2018) 86:841-859

\bibitem{Ding2015}
Y. Ding, On list-decodability of random rank-metric codes and subspace codes,  {\it IEEE Transactions on Information Theory}, vol. 61, no. 1, pp 51-59, 2015.

\bibitem{Z2012}
Z. Dvir and S. Lovett, subspace evasive sets,  {\it Proceedings of the 44th ACM Symposium on Theory of Computing,} pp: 351- 358,2012.

\bibitem{P.E1957}
P. Elias, List decoding for noisy channels, Research Laboratory of Electronics, Massachusetts Institute of Technology, 1957.

\bibitem{thesis}
V. Guruswami, List Decoding of Error-Correcting Codes, Springer, US, 2001.

\bibitem{V.G2008}
V. Guruswami and A. Rudra, Explicit codes achieving list decoding capacity: Error-correction with optimal redundancy, {\it IEEE Transactions on Information Theory}, vol. 54, no. 1, pp 135-150, 2008.


\bibitem{Gur2005}
V. Guruswami and S. Vadhan, A low bound on list size for list decoding,  {\it IEEE Transactions on Information Theory}, vol. 56, no. 11, pp 5681-5688, 2010.

\bibitem{VGX2012}
V. Guruswami and C. Xing, List decoding Reed-Solomon, Algebraic Geometric, and Gabidulin subcodes up to the Singleton bound, in {\it Electronic Colloquium on Computational Complexity (ECCC),} 19:146, 2012. Extended abstract appeared in the Proceedings of the 45th ACM Symposium on Theory of Computing (STOC'13).

\bibitem{Kai}
X. Kai, S. Zhu, P. Li, A construction of new MDS symbol-pair codes, {\it IEEE Transactions on Information Theory}, 61(11),
5828-5834 (2015).

\bibitem{Makoto2016}
M. Takita, M. Hirotomo and M. Morii, Error-Trapping decoding for cyclic codes over symbol-pair read channels, {\it International Symposium on Information Theory and Its Applications,} pp. 681-685, California, USA, 2016.

\bibitem{Masanori2014}
M. Hirotomo, M. Takita and M. Morii, Syndrome decoding of symbol-pair codes, {\it IEEE Information Theory Workshop}, pp. 162-166, Australia, 2014.

\bibitem{Shunsuke2016}
S. Horii, T. Matsushima and S. Hirasawa, Linear Programming decoding of binary linear codes for symbol-pair read channels, in {\it IEEE International Symposium on Information Theory,} pp. 1944-1948, Spain, 2016.

\bibitem{Liu2018}
S. Liu, C. Xing and C. Yuan, List decoding of cover-metric codes up to the Singleton bound,  {\it IEEE Transactions on Information Theory}, vol. 64, no. 4, pp 2410-2416, 2018.

\bibitem{J.M.W1958}
J. M. Wozencraft, List decoding, Quarterly Progress Report, Research Laboratory of Electronics, MIT, 48, pp. 90-95, 1958.

\bibitem{Yaa2016}
E. Yaakobi, J. Bruck and P. H. Siegel, Constructions and decoding of cyclic codes over $b$-symbol read channels, {\it IEEE Transactions on Information Theory}, vol. 62, no. 4, pp 1541-1551.

\bibitem{Yaa2012}
E. Yaakobi, J. Bruck and P. H. Siegel, Decoding of cyclic codes over symbol-pair read channels,  {\it IEEE International Symposium on Information Theory,} Cambridge, MA, pp. 2891-2895, 2012.

\end{thebibliography}
\end{document}